\def\llncs{0}
\def\fullpage{1}
\def\anonymous{0}
\def\authnote{0}
\def\notxfont{0}
\def\submission{0}
\def\reply{0}
\def\cameraready{0}
\def\noaux{0}
\def\displaylabel{0}
\ifnum\submission=1
\def\anonymous{1}
\def\llncs{1}
\def\displaylabel{0}
\else
\fi

\ifnum\cameraready=1
\def\llncs{1}
\def\anonymous{0}
\def\authnote{0}
\def\displaylabel{0}
\else
\fi

\ifnum\anonymous=1
\def\authnote{0}
\else
\fi

\ifnum\llncs=1
	\documentclass[envcountsect,a4paper,ruhttps://www.overleaf.com/project/62afd95023a007072d4588bfnningheads,10pt]{llncs}
\else
	\documentclass[letterpaper,hmargin=1.05in,vmargin=1.05in]{article}
			\ifnum\fullpage=1
		\usepackage{fullpage}
		\fi
\fi

\ifnum\reply=1
\usepackage{ulem}
\renewcommand{\emph}{\textit}
\else
\fi

\ifnum\reply=0  

\newcommand{\strike}[1]{}
\else

\newcommand{\strike}[1]{\sout{#1}}
\fi

\usepackage[%
  colorlinks=true,
  citecolor=darkgreen,
  linkcolor=darkblue,
  urlcolor=darkblue,
  pagebackref=true
]{hyperref}

\usepackage{amsmath, amsfonts, amssymb, mathtools,amscd}

\usepackage{amsthm}
\usepackage{mathrsfs}

\usepackage{lmodern}
\usepackage[T1]{fontenc}
\usepackage[utf8]{inputenc}

\usepackage{etex}

\usepackage{arydshln} 
\usepackage{url}
\usepackage{ifthen}
\usepackage{bm}
\usepackage{multirow}
\usepackage[dvips]{graphicx}
\usepackage[usenames]{color}
\usepackage{xcolor,colortbl} 
\usepackage{threeparttable}
\usepackage{comment}
\usepackage{paralist,verbatim}
\usepackage{cases}
\usepackage{booktabs}
\usepackage{braket}
\usepackage{cancel} 
\usepackage{ascmac} 
\usepackage{framed}
\usepackage{authblk}
\usepackage{pifont}
\usepackage{physics}
\definecolor{darkblue}{rgb}{0,0,0.6}
\definecolor{darkgreen}{rgb}{0,0.5,0}
\definecolor{maroon}{rgb}{0.5,0.1,0.1}
\definecolor{dpurple}{rgb}{0.2,0,0.65}
\definecolor{chocolate}{rgb}{0.8,0.4,0.1}

\usepackage[capitalise,noabbrev]{cleveref}
\usepackage[absolute]{textpos}
\usepackage[final]{microtype}
\usepackage[absolute]{textpos}
\usepackage{autonum}

\usepackage{dsfont}
\DeclareMathAlphabet{\mathpzc}{OT1}{pzc}{m}{it}

\ifnum\submission=1
\renewcommand*{\backref}[1]{}
\def\notxfont{1}
\else
\fi

\ifnum\llncs=1
\renewcommand{\subparagraph}{\paragraph}
\else
\ifnum\notxfont=1
\else
\usepackage{mathpazo}
\usepackage{newtxtext}
\usepackage{helvet}
\fi
\fi

\newtheoremstyle{thicktheorem}%
{\topsep}
{\topsep}
{\itshape}{}%
{\bfseries}%
{.}
{ }%
{\thmname{#1}\thmnumber{ #2}%
		\thmnote{ (#3)}%
}

\newtheoremstyle{remark}
{\topsep}
{\topsep}
	{}
	{}
	{}
	{.}
	{ }
	{\textit{\thmname{#1}}\thmnumber{ #2}
			\thmnote{ (#3)}%
	}

\ifnum\llncs=0
	\theoremstyle{thicktheorem}
	\newtheorem{theorem}{Theorem}[section]
	\newtheorem{lemma}[theorem]{Lemma}

	\newtheorem{definition}[theorem]{Definition}

	\theoremstyle{remark}
	
	\newtheorem{remark}[theorem]{Remark}

\fi

	\crefname{theorem}{Theorem}{Theorems}
	\crefname{assumption}{Assumption}{Assumptions}
	\crefname{construction}{Construction}{Constructions}
	\crefname{corollary}{Corollary}{Corollaries}
	\crefname{conjecture}{Conjecture}{Conjectures}
	\crefname{definition}{Definition}{Definitions}
	\crefname{exmaple}{Example}{Examples}
	\crefname{experiment}{Experiment}{Experiments}
	\crefname{counterexample}{Counterexample}{Counterexamples}
	\crefname{lemma}{Lemma}{Lemmata}
	\crefname{observation}{Observation}{Observations}
	\crefname{proposition}{Proposition}{Propositions}
	\crefname{remark}{Remark}{Remarks}
	\crefname{claim}{Claim}{Claims}
	\crefname{fact}{Fact}{Facts}
	\crefname{note}{Note}{Notes}

\ifnum\llncs=1
 \crefname{appendix}{App.}{Appendices}
 \crefname{section}{Sec.}{Sections}
\else
\fi

\ifnum\llncs=1
\renewcommand*{\backref}[1]{}
\else
	\renewcommand*{\backref}[1]{(Cited on page~#1.)}
	\ifnum\notxfont=1
	\else
		\usepackage{newtxtext}
	\fi
\fi

\usepackage{etoolbox}
\newcommand*{\keys}[1]{\mathsf{#1}}

\newcommand*{\algo}[1]{\ensuremath{\mathsf{#1}}}
\newcommand*{\qalgo}[1]{\ensuremath{\mathpzc{#1}}}
\newcommand*{\qstate}[1]{\mathpzc{#1}}
\newcommand*{\qreg}[1]{{\color{gray}{\mathsf{#1}}}}

\newcounter{expitem}

\usepackage{mathtools}

\newcommand{\chosen}{\leftarrow}

\newcommand{\lrun}{\leftarrow}

\newcommand{\la}{\leftarrow}
\newcommand{\ra}{\rightarrow}
\renewcommand{\gets}{\leftarrow}

\newcommand{\seteq}{\coloneqq}

\newcommand{\tensor}{\otimes}

\newcommand{\setbk}[1]{\{#1\}}

\newcommand{\cM}{\mathcal{M}}

\newcommand{\cZ}{\mathcal{Z}}

\newcommand{\qA}{\qalgo{A}}
\newcommand{\qB}{\qalgo{B}}

\def\makeuppercase#1{
\expandafter\newcommand\csname sf#1\endcsname{\mathsf{#1}}
\expandafter\newcommand\csname frak#1\endcsname{\mathfrak{#1}}
\expandafter\newcommand\csname bb#1\endcsname{\mathbb{#1}}
\expandafter\newcommand\csname bf#1\endcsname{\textbf{#1}}
}

\def\makelowercase#1{
\expandafter\newcommand\csname frak#1\endcsname{\mathfrak{#1}}
\expandafter\newcommand\csname bf#1\endcsname{\textbf{#1}}
}

\newcounter{char}
\loop
   \edef\letter{\alph{char}}
   \edef\Letter{\Alph{char}}
   \expandafter\makelowercase\letter
   \expandafter\makeuppercase\Letter
   \stepcounter{char}
   \unless\ifnum\thechar>26
\repeat

\def\makeuppercase#1{
\expandafter\newcommand\csname tl#1\endcsname{\widetilde{#1}}
}

\def\makelowercase#1{
\expandafter\newcommand\csname tl#1\endcsname{\widetilde{#1}}
}

\newcommand{\N}{\mathbb{N}}

\newcommand{\R}{\mathbb{R}}

\newcommand{\bit}{\{0,1\}}

\newcommand{\Ms}{\mathcal{M}}

\newcommand{\secp}{\lambda}

\newcommand{\cert}{\keys{cert}}

\newcommand{\advb}[3]{\mathsf{Adv}_{#1}^{\mathsf{#2} \mbox{-} \mathsf{#3}}}

\newcommand{\expb}[3]{\mathsf{Exp}_{#1}^{ \mathsf{#2} \mbox{-} \mathsf{#3}}}

\newcommand*{\pk}{\keys{pk}}
\newcommand*{\sk}{\keys{sk}}

\newcommand*{\vk}{\keys{vk}}

\newcommand*{\ct}{\keys{ct}}

\newcommand*{\msg}{m}

\newcommand{\qct}{\qstate{ct}}

\newenvironment{boxfig}[2]{\begin{figure}[#1]\fbox{\begin{minipage}{0.97\linewidth}
                        \vspace{0.2em}
                        \makebox[0.025\linewidth]{}
                        \begin{minipage}{0.95\linewidth}
            {{
                        #2 }}
                        \end{minipage}
                        \vspace{0.2em}
                        \end{minipage}}
                        }
                        {\end{figure}}

\newcommand{\Gen}{\algo{Gen}}

\newcommand{\keygen}{\algo{KeyGen}}

\newcommand{\Enc}{\algo{Enc}}
\newcommand{\Dec}{\algo{Dec}}

\newcommand{\enc}{\algo{Enc}}

\newcommand{\Sign}{\algo{Sign}}
\newcommand{\Vrfy}{\algo{Vrfy}}

\newcommand{\qEnc}{\qalgo{Enc}}
\newcommand{\qDec}{\qalgo{Dec}}
\newcommand{\qDelete}{\qalgo{Del}}

\newcommand\PKE{\algo{PKE}}

\newcommand\SIG{\algo{SIG}}

\newcommand{\negl}{{\mathsf{negl}}}

\newcommand{\zo}[1]{\{0,1\}^{#1}}

\newcommand{\xor}{\oplus}

\newcommand{\class}[1]{\mathsf{#1}}

\newcommand{\NP}{\class{NP}}

\newcommand{\sigklen}{\ell_{\sigk}}
\newcommand{\ttForge}{\mathtt{Forge}}

\newcommand{\pke}{\mathsf{pke}}
\newcommand{\PKEPVD}{\mathsf{PKE}\textrm{-}\mathsf{PVD}}
\newcommand{\sigk}{\mathsf{sigk}}

\newcommand{\msf}[1]{\mathsf{#1}}

\makeatletter
\newcommand{\vast}{\bBigg@{3.5}}
\newcommand{\Vast}{\bBigg@{5}}
\makeatother

\newcommand*{\TD}{\msf{TD}}
\usepackage{fancyhdr}
\ifnum\displaylabel=1
\usepackage{showkeys}
\else
\fi

\ifnum\llncs=1
\let\oldvec\vec
\let\vec\oldvec
\makeatletter
\renewcommand*\l@author[2]{}
\renewcommand*\l@title[2]{}
\makeatletter
\fi    

\theoremstyle{remark}

\definecolor{darkblue}{rgb}{0,0,0.6}
\definecolor{darkgreen}{rgb}{0,0.5,0}
\definecolor{maroon}{rgb}{0.5,0.1,0.1}
\definecolor{dpurple}{rgb}{0.2,0,0.65}
\definecolor{darkviolet}{RGB}{130,95,141}
\definecolor{darkkhaki}{RGB}{189,183,107}

\ifnum\authnote=1
\newcommand{\ryo}[1]{\textcolor{brown}{[{\footnotesize {\bf RN:} { {#1}}}]}}
\newcommand{\fuyuki}[1]{\textcolor{blue}{[{\footnotesize {\bf FK:} { {#1}}}]}}
\newcommand{\takashi}[1]{\textcolor{dpurple}{[{\footnotesize {\bf TY:} { {#1}}}]}}
\else
\newcommand{\ryo}[1]{}
\newcommand{\fuyuki}[1]{}
\newcommand{\takashi}[1]{}
\fi

\title{
\textbf{Publicly Verifiable Deletion from Minimal Assumptions}
\ifnum\noaux=1
\else
\ifnum\submission=1
\thanks{{\color{red}{\emph{We attached the full version of this paper as a supplementary material}}}.}
\else
\fi
\fi
}

\begin{document}

\ifnum\anonymous=1
\author{\empty}
\ifnum\llncs=1
\institute{\empty}
\else
\fi
\else
%
%
\ifnum\llncs=1
\author{
Fuyuki Kitagawa\inst{1}\and Ryo Nishimaki\inst{1} \and Takashi Yamakawa\inst{1}
}
\institute{
NTT Social Informatics Laboratories, 
}
\else
%
%

\author[$\dagger$]{Fuyuki Kitagawa}
\author[$\dagger$]{Ryo Nishimaki}
\author[$\dagger$]{Takashi Yamakawa}
\affil[$\dagger$]{{\small NTT Social Informatics Laboratories, Tokyo, Japan}\authorcr{\small \{fuyuki.kitagawa,ryo.nishimaki,takashi.yamakawa\}@ntt.com}}
\renewcommand\Authands{, }
\fi 
\fi

\ifnum\llncs=1
\date{}
\else
\date{\today}
\fi

\maketitle

\begin{abstract}
We present a general compiler to add the publicly verifiable deletion property for various cryptographic primitives including public key encryption, attribute-based encryption, and quantum fully homomorphic encryption. Our compiler only uses one-way functions, or more generally hard quantum planted problems for $\NP$, which are implied by one-way functions. 
It relies on minimal assumptions and enables us to add the publicly verifiable deletion property with no additional assumption for the above primitives. Previously, such a compiler needs additional assumptions such as injective trapdoor one-way functions or pseudorandom group actions [Bartusek-Khurana-Poremba, ePrint:2023/370]. Technically,  we upgrade an existing compiler for privately verifiable deletion [Bartusek-Khurana, ePrint:2022/1178] to achieve publicly verifiable deletion by using digital signatures. 
\end{abstract}

\ifnum\llncs=1
\else
\newpage
\setcounter{tocdepth}{2}
\tableofcontents

\newpage
\fi


\newcommand{\txPVD}{\textrm{PVD}}
\newcommand{\sfAoNE}{\mathsf{AoE}}

\section{Introduction}\label{sec:intro}

\subsection{Background}\label{sec:background}
Quantum mechanics yields new cryptographic primitives that classical cryptography cannot achieve.
In particular, the uncertainty principle enables us to certify deletion of information. 
Broadbent and Islam~\cite{TCC:BroIsl20} introduced the notion of quantum encryption with certified deletion, where we can generate a classical certificate for the deletion of quantum ciphertext. We need a verification key generated along with a quantum ciphertext to check the validity of a certificate. The root of this concept is revocable quantum time-release encryption by Unruh~\cite{JACM:Unruh15}, where a sender can revoke quantum ciphertext if a receiver returns it before a pre-determined time.  Encryption with certified deletion is useful because encryption security holds \emph{even if adversaries obtain a secret decryption key after they generate a valid certificate for deletion}.
After the work by Broadbent and Islam~\cite{TCC:BroIsl20}, many works presented extended definitions and new constructions of quantum (advanced) encryption with certified deletion~\cite{AC:HMNY21,ITCS:Poremba23,EPRINT:BarKhu22,EPRINT:HKMNPY23,EPRINT:BGGKMRR23,myEPRINT:BarKhuPor23}. In particular, Bartusek and Khurana~\cite{EPRINT:BarKhu22}, and Hiroka, Kitagawa, Morimae, Nishimaki, Pal, and Yamakawa~\cite{EPRINT:HKMNPY23} considered certified everlasting security, which guarantees that computationally \emph{unbounded} adversaries cannot obtain information about plaintext \emph{after a valid certificate was generated}. Several works~\cite{AC:HMNY21,ITCS:Poremba23,EPRINT:BGGKMRR23,myEPRINT:BarKhuPor23} considered public verifiability, where we can reveal verification keys without harming certified deletion security. These properties are desirable for encryption with certified deletion in real-world applications.

In this work, we focus on cryptographic primitives with publicly verifiable deletion ($\txPVD$)~\cite{myEPRINT:BarKhuPor23}, which satisfy certified everlasting security and public verifiability. Known schemes based on BB84 states are privately verifiable~\cite{TCC:BroIsl20,AC:HMNY21,EPRINT:BarKhu22,EPRINT:HKMNPY23}, where we need to keep verification keys secret to ensure encryption security.
Some schemes are publicly verifiable~\cite{AC:HMNY21,ITCS:Poremba23,EPRINT:BGGKMRR23,myEPRINT:BarKhuPor23}. Public verifiability is preferable to private verifiability since we need to keep many verification keys secret when we generate many ciphertexts of encryption with privately verifiable deletion. More secret keys lead to more risks. In addition, anyone can verify deletion in cryptography with $\txPVD$. 

Hiroka, Morimae, Nishimaki, and Yamakawa~\cite{AC:HMNY21} achieved interactive public key encryption with non-everlasting $\txPVD$  from extractable witness encryption~\cite{C:GKPVZ13}, which is a strong knowledge-type assumption, and one-shot signatures which require classical oracles~\cite{STOC:AGKZ20}.
Poremba~\cite{ITCS:Poremba23} achieved public key encryption (PKE) and fully homomorphic encryption (FHE) with non-everlasting $\txPVD$ based on lattices. He conjectured Ajtai hash function satisfies a strong Gaussian collapsing property and proved the security of his constructions under the conjecture. Later, Bartusek, Khurana, and Poremba~\cite{myEPRINT:BarKhuPor23} proved the conjecture is true under the LWE assumption.

Bartusek, Garg, Goyal, Khurana, Malavolta, Raizes, and Roberts~\cite{EPRINT:BGGKMRR23} achieved primitive $X$ with $\txPVD$ from $X$ and post-quantum secure indistinguishability obfuscation (IO)~\cite{JACM:BGIRSVY12} where $X \in \setbk{\textrm{SKE, COM, PKE, ABE, FHE, TRE, WE}}$.\footnote{SKE, COM, ABE, TRE, and WE stand for secret key encryption, commitment, attribute-based encryption, time-release encryption, and witness encryption, respectively. Although Bartusek et al.~\cite{EPRINT:BGGKMRR23} did not mention, we can apply their transformation to SKE and COM as the results by Bartusek and Khurana~\cite{EPRINT:BarKhu22}.} They also achieved functional encryption with $\txPVD$ and obfuscation with $\txPVD$, which rely on post-quantum IO. All their constructions use subspace coset states~\cite{C:CLLZ21}. 
Bartusek et al.~\cite{myEPRINT:BarKhuPor23} achieved PKE (resp. COM) with $\txPVD$ from injective trapdoor one-way functions (or superposition-invertible regular trapdoor functions) or pseudorandom group actions~\cite{EPRINT:HhaMorYam22} (resp. almost-regular one-way functions). They also achieved primitive $Y$ with $\txPVD$ from injective trapdoor one-way functions (or superposition-invertible regular trapdoor functions) or pseudorandom group actions, and primitive $Y$ where $Y\in \setbk{\textrm{ABE, QFHE, TRE, WE}}$. They obtained these results by considering certified everlasting target-collapsing functions as an intermediate primitive.

As we explained above, known encryption with $\txPVD$ constructions need strong and non-standard assumptions~\cite{AC:HMNY21,EPRINT:BGGKMRR23}, algebraic assumptions~\cite{ITCS:Poremba23,myEPRINT:BarKhuPor23}, or additional assumptions~\cite{myEPRINT:BarKhuPor23}. This status is unsatisfactory because Bartusek and Khurana~\cite{EPRINT:BarKhu22} prove that we can achieve $X$ with \emph{privately} verifiable deletion\footnote{We do not abbreviate when we refer to this type to avoid confusion.} from $X$ where $X \in \setbk{\textrm{SKE, COM, PKE, ABE, FHE, TRE, WE}}$\footnote{Although Bartusek and Khurana~\cite{EPRINT:BarKhu22} did not mention, we can apply their transformation to SKE.} \emph{without any additional assumptions}. Thus, our main question in this work is the following.

\begin{center}
\emph{Can we achieve encryption with $\txPVD$ from minimal assumptions?}
\end{center}

\subsection{Our Results}\label{sec:results} 
We affirmatively answer the main question described in the previous section.
We present a general transformation from primitive $Z$ into $Z$ with $\txPVD$ where $Z \in \setbk{\textrm{SKE, COM, PKE, ABE, QFHE, TRE, WE}}$. 
In the transformation, we only use one-way functions, or more generally, hard quantum planted problems for $\NP$, which we introduce in this work and are implied by one-way functions. 

More specifically, we extend the certified everlasting lemma by Bartusek and Khurana~\cite{EPRINT:BarKhu22}, which enables us to achieve encryption with privately verifiable deletion, to a publicly verifiable certified everlasting lemma.
We develop an authentication technique based on (a variant of) the Lamport signature~\cite{Lamport79} to reduce our publicly verifiable certified everlasting lemma to Bartusek and Khurana's certified everlasting lemma. 

Our new lemma is almost as versatile as Bartusek and Khurana's lemma. It is easy to apply it to all-or-nothing type encryption\footnote{SKE, PKE, ABE, (Q)FHE, TRE, and WE fall into this category.}, where a secret key (or witness) holder can recover the entire plaintext. One subtle issue is that we need to use QFHE for (Q)FHE with $\txPVD$. The reason is that we need to apply an evaluation algorithm to quantum ciphertext. Note that we can use FHE and Bartusek and Khurana's lemma to achieve FHE with privately verifiable deletion.

The advantages of our techniques are as follows:
\begin{itemize}
\item For $Z' \in \setbk{\textrm{SKE, COM, PKE, ABE, QFHE, TRE}}$, we can convert plain $Z'$ into $Z'$ with $\txPVD$ with no additional assumption.  
For WE, we can convert it into WE with $\txPVD$ additionally assuming one-way functions (or hard quantum planted problems for $\NP$).\footnote{WE does not seem to imply one-way functions.}  
Bartusek et al.~\cite{myEPRINT:BarKhuPor23} require injective (or almost-regular) trapdoor one-way functions, pseudorandom group actions, or almost-regular one-way functions for their constructions.
\item 
Our transformation is applicable even if the base scheme has quantum encryption and decryption (or committing) algorithms.\footnote{The compilers of \cite{EPRINT:BarKhu22,myEPRINT:BarKhuPor23} are also applicable to schemes that have quantum encryption and decryption (or committing) algorithms though they do not explicitly mention it. \takashi{Is this correct?}\fuyuki{I think so.}\ryo{Regarding~\cite{EPRINT:BarKhu22}, it is correct. Regarding~\cite{myEPRINT:BarKhuPor23}, this issue is subtle. More precisely, they directly construct PKE and COM with $\txPVD$, not from classical PKE and COM via a compiler.}} 
This means that our assumptions are minimal since  $Z$ with $\txPVD$ implies both plain $Z$ with quantum encryption and decryption (or committing) algorithms and hard quantum planted problems for $\NP$ 
for $Z \in \setbk{\textrm{SKE, COM, PKE, ABE, QFHE, TRE, WE}}$. 

\item Our approach is simple. Bartusek et al.~\cite{myEPRINT:BarKhuPor23} introduced an elegant intermediate notion, certified everlasting target-collapsing, to achieve cryptography with $\txPVD$. However, they use a few more intermediate notions (such as balanced binary-measurement target-collision-resistance) for their approach.
\end{itemize}


\subsection{Technical Overview}\label{sec:tech_overview}
As explained in \Cref{sec:background}, Bartusek and Khurana~\cite{EPRINT:BarKhu22} gave a generic compiler to add the \emph{privately verifiable} deletion property for various types of encryption. 
Our finding is that there is a surprisingly simple way to upgrade their compiler to achieve \emph{publicly verifiable} deletion by additionally using digital signatures. 

\paragraph{Notations.}
For a bit string $x$, we write $x_j$ to mean $j$-th bit of $x$. 
For bit strings $x,\theta\in \bit^{\ell}$, we write $\ket{x}_\theta$ to mean the BB84 state $\bigotimes_{j\in [\ell]}H^{\theta_j}\ket{x_j}$ where $H$ is the Hadamard operator. 
 
\paragraph{Certified everlasting lemma of \cite{EPRINT:BarKhu22}.}
The compiler of  \cite{EPRINT:BarKhu22} is based on the \emph{certified everlasting lemma} described below.\footnote{For simplicity, we state a simplified version of the lemma that is sufficient for the conversion for PKE, FHE, TRE, and WE, but not for ABE.
See \Cref{lem:ce} for the general version.} We describe it in the dual version, where the roles of computational and Hadamard bases are swapped for convenience. We stress that the dual version is equivalent to the original one because they coincide under an appropriate basis change.

Consider a family of distributions $\{{\mathcal{Z}(m)}\}_{m\in\bit^{\secp+1}}$ over classical strings, such that for any $m\in \bit^{\secp+1}$, the distribution $\mathcal{Z}(m)$ is computationally indistinguishable from the distribution $\mathcal{Z}(0^{\secp+1})$. In other words, each distribution $\mathcal{Z}(m)$ can be thought of as an ``encryption'' of the input $m$.
Let $\widetilde{\mathcal{Z}}(b)$ be an experiment between an adversary and challenger defined as follows for $b\in \bit$: 
\begin{itemize}
    \item The challenger samples $x, \theta \leftarrow \bit^{\secp}$ and sends $\ket{x}_{\theta}$ and $\mathcal{Z}(\theta,b \oplus \bigoplus_{j: \theta_j = 1} x_j)$ to the adversary. 
    \item The adversary sends a classical string $x'\in \bit^\secp$ and a quantum state $\rho$ to the challenger. 
    \item The challenger outputs $\rho$ if $x'_j=x_j$ for all $j$ such that $\theta_j=0$, and otherwise outputs a special symbol $\bot$ as the output of the experiment.
\end{itemize} 
The certified everlasting lemma states that 
for any QPT adversary, the trace distance between $\widetilde{\mathcal{Z}}(0)$ and $\widetilde{\mathcal{Z}}(1)$ is negligible in $\secp$.

The above lemma immediately gives a generic compiler to add the privately verifiable deletion property.  To encrypt a message $b\in \bit$,  we set the ciphertext to be $(\ket{x}_{\theta},\enc(\theta,b \oplus \bigoplus_{j: \theta_j = 1} x_j))$, where $x$ and $\theta$ are randomly chosen from $\bit^\secp$, and $\enc$ is the encryption algorithm of the underlying scheme.\footnote{We write $\enc(\theta,b \oplus \bigoplus_{j: \theta_j = 1} x_j)$ to mean an encryption of the message $(\theta,b \oplus \bigoplus_{j: \theta_j = 1} x_j)$ where we omit the encryption key.} 
The decryptor first decrypts the second component to get $(\theta,b \oplus \bigoplus_{j: \theta_j = 1} x_j)$, recovers $ \bigoplus_{j: \theta_j = 1} x_j$ from $\ket{x}_{\theta}$ and $\theta$, and then XORs it with $b \oplus \bigoplus_{j: \theta_j = 1} x_j$ to recover $b$.  
To delete the ciphertext and obtain a certificate $x'$, we measure $\ket{x}_{\theta}$ in the standard basis.
To verify the certificate, we check if $x'_j=x_j$ for all $j$ such that $\theta_j=0$. 
By utilizing the above lemma, we can see that an adversary's internal state will not contain any information about $b$ given the verification algorithm's acceptance. Therefore, the scheme offers certified everlasting security. However, this scheme has only privately verifiable deletion because the verification algorithm needs to know $x$ and $\theta$, which are part of the encryption randomness that has to be hidden from the adversary. 

\paragraph{Making verification public via digital signatures.} 
We show a publicly verifiable variant of the certified everlasting lemma by using digital signatures. Roughly speaking, our idea is to generate a signature for the BB84 state $\ket{x}_{\theta}$ by coherently running the signing algorithm so that the verification of deletion can be done by the verification of the signature, which can be done publicly. Note that the signature does \emph{not} certify $\ket{x}_{\theta}$ as a quantum state. It rather certifies its computational basis part (i.e., $x_j$ for $j$ such that $\theta_j=0$). This is sufficient for our purpose because the verification of deletion just checks the computational basis part.  

With the above idea in mind, we modify the experiment $\widetilde{\mathcal{Z}}(b)$ as follows:
\begin{itemize}
    \item 
    The challenger 
    generates a key pair $(\vk,\sigk)$ of a digital signature scheme and
    samples $x, \theta \leftarrow \bit^{\secp}$. 
    Let $U_{\mathsf{sign}}$ be a unitary that works as follows:
  \[
        \ket{m}\ket{0\ldots0}\mapsto\ket{m}\ket{\Sign(\sigk,m)}.
        \]
    where $\Sign(\sigk,\cdot)$ is a deterministic signing algorithm with a signing key $\sigk$. 
    The challenger sends $\vk$, $U_{\mathsf{sign}} \ket{x}_{\theta}\ket{0\ldots 0}$, and $\mathcal{Z}(\sigk,\theta,b \oplus \bigoplus_{j: \theta_j = 1} x_j)$ to the adversary. 
    \item The adversary sends a classical string $x'\in \bit^\secp$, a signature $\sigma$, and a quantum state $\rho$ to the challenger. 
    \item The challenger outputs $\rho$ if $\sigma$ is a valid signature for $x'$, and otherwise outputs a special symbol $\bot$ as the output of the experiment.
\end{itemize} 
We show that for any QPT adversary, the trace distance between $\widetilde{\mathcal{Z}}(0)$ and $\widetilde{\mathcal{Z}}(1)$ is negligible in $\secp$ if we instantiate the digital signatures  with an appropriate scheme as explained later.
The crucial difference from the original lemma is that the challenger does not need to check if $x'_j=x_j$ for all $j$ such that $\theta_j=0$ and only needs to run the verification algorithm of the digital signature scheme. 

By using the above variant similarly to the original privately verifiable construction, we obtain a generic compiler that adds publicly verifiable deletion property. For clarity, we describe the construction below. 
To encrypt a message $b\in \bit$, the encryption algorithm generates a key pair $(\vk,\sigk)$ of the digital signature scheme, chooses $x,\theta\gets \bit^\secp$, and outputs a ciphertext $(U_{\mathsf{sign}}\ket{x}_{\theta},\enc(\sigk, \theta,b \oplus \bigoplus_{j: \theta_j = 1} x_j))$ and a public verification key $\vk$.  
The decryptor first decrypts the second component to get $(\sigk,\theta,b \oplus \bigoplus_{j: \theta_j = 1} x_j)$, uncompute the signature register of the first component by using $\sigk$ to get $\ket{x}_{\theta}$, recovers $ \bigoplus_{j: \theta_j = 1} x_j$ from $\ket{x}_{\theta}$ and $\theta$, and then XORs it with $b \oplus \bigoplus_{j: \theta_j = 1} x_j$ to recover $b$.  
To delete the ciphertext and obtain a certificate $(x',\sigma)$, we measure $U_{\mathsf{sign}}\ket{x}_{\theta}$ in the standard basis to get $(x',\sigma)$.
To verify the certificate, we check if $\sigma$ is a valid signature for $x'$.  
By utilizing the above lemma, we can see that the above scheme achieves certified everlasting security with public verification.

\paragraph{Proof idea and instantiation of digital signatures.} 
We prove the above publicly verifiable version by reducing it to the original one in  \cite{EPRINT:BarKhu22}. 
Noting that $\mathcal{Z}(\sigk,\theta,b \oplus \bigoplus_{j: \theta_j = 1} x_j)$ computationally hides $\sigk$ by the assumption, a straightforward reduction works if the digital signature scheme satisfies the following security notion which we call  one-time unforgeability for BB84 states.
\begin{definition}[One-time unforgeability for BB84 states (informal)]
Given $\vk$ and $U_{\mathsf{sign}}\ket{x}_{\theta}\ket{0\ldots 0}$ for uniformly random $x,\theta \gets \bit^\secp$, no QPT adversary can output $x'\in \bit^\secp$ and a signature $\sigma$ such that  $\sigma$ is a valid signature for $x'$ 
and $x'_j\neq x_j$ for some $j$ such that $\theta_j=0$ with a non-negligible probability.  
\end{definition}  
It is easy to show that the Lamport signature satisfies the above property. This can be seen as follows. 
Recall that a verification key of the Lamport signature consists of $(v_{j,b})_{j\in [\secp],b\in \bit}$ where 
$v_{j,b}\seteq f(u_{j,b})$ for a one-way function $f$ and uniformly random inputs  $(u_{j,b})_{j\in [\secp],b\in \bit}$, and a signature for a message $m\in \bit^\secp$ is $\sigma\seteq (u_{j,m_j})_{j\in [\secp]}$.  
Suppose that there is an adversary that breaks the above property. Then, there must exist $j\in [\secp]$ such that $x'_j\neq x_j$ and $\theta_j=0$, in which case the input state $U_{\mathsf{sign}}\ket{x}_{\theta}\ket{0\ldots 0}$ does not have any information of $u_{j,1-x_j}$. On the other hand, to generate a valid signature for $x'$, the adversary has to find a preimage of  $v_{j,x'_j}=v_{j,1-x_j}$.  
This is impossible by the one-wayness of $f$. Thus, a digital signature scheme that satisfies one-time unforgeability for BB84 states exists assuming the existence of one-way functions.  

\paragraph{Achieving minimal assumptions.}
In the above, we explain that one-way functions are sufficient for instantiating the digital signature scheme needed for our compiler. On the other hand, encryption schemes with publicly verifiable deletion does not seem to imply the existence of one-way functions because ciphertexts can be quantum. Thus, one-way functions may not be a necessary assumption for them. To weaken the assumption, we observe that we can use hard quantum planted problems for $\NP$ instead of one-way functions in the Lamport signature. Here, hard quantum planted problems for a $\NP$ language $L$ is specified by a quantum polynomial-time sampler that samples an instance-witness pair $(x,w)$ for the language $L$ in such a way that no QPT adversary can find a witness for $x$  with non-negligible probability. Given such a sampler, it is easy to see that we can instantiate the digital signature scheme similar to the above except that $(v_{j,b},u_{j,b})$ is now replaced with an instance-witness pair sampled by the sampler. 

We observe that $Z$ with $\txPVD$ implies the existence of hard quantum planted problems for $\NP$ where $Z \in \setbk{\textrm{SKE, COM, PKE, ABE, QFHE, TRE, WE}}$.\footnote{We assume that the verification algorithm of $Z$ with $\txPVD$ is a classical deterministic algorithm. If we allow it to be a quantum algorithm, we have to consider hard quantum planted problems for $\mathsf{QCMA}$, which are also sufficient to instantiate our compiler.} This can be seen by considering the verification key as an instance and certificate as a witness for an $\NP$ language. Our construction relies on hard quantum planted problems for $\NP$  and plain $Z$ with quantum encryption and decryption (or committing) algorithms, both of which are implied by $Z$ with $\txPVD$, and thus it is based on the minimal assumptions.

\subsection{More on Related Works}\label{sec:related_works}

\paragraph{Certified deletion for ciphertext with private verifiability.}
Broadbent and Islam~\cite{TCC:BroIsl20} achieved one-time SKE with privately verifiable deletion without any cryptographic assumptions.
Hiroka et al.~\cite{AC:HMNY21} achieved PKE and ABE with non-everlasting privately verifiable deletion. They also achieve interactive PKE with non-everlasting privately verifiable deletion and classical communication in the quantum random oracle model (QROM) from the LWE assumption. However, none of their constructions satisfy certified everlasting security.
Hiroka, Morimae, Nishimaki, and Yamakawa~\cite{C:HMNY22} defined certified everlasting commitment and zero-knowledge (for $\class{QMA}$) by extending everlasting security~\cite{JC:MulUnr10}. They achieved these notions by using plain commitment and QROM.
Bartusek and Khurana~\cite{EPRINT:BarKhu22} defined certified everlasting security for various all-or-nothing type encryption primitives and presented a general framework to achieve certified everlasting secure all-or-nothing type encryption and commitment with privately verifiable deletion using BB84 states. They also studied secure computation with everlasting security transfer.
Hiroka et al.~\cite{EPRINT:HKMNPY23} also defined and achieved certified everlasting security for various cryptographic primitives. In particular, they defined and achieved certified everlasting (collusion-resistant) functional encryption, garbled circuits, and compute-and-compare obfuscation, which are outside of all-or-nothing type encryption. Their constructions are also based on BB84 states and are privately verifiable. They also use a signature-based authentication technique for BB84 states. However, its role is not achieving public verifiability but functional encryption security. See the paper by Hiroka et al.~\cite{EPRINT:HKMNPY23} for the differences between their results and Bartusek and Khurana's results~\cite{EPRINT:BarKhu22}.  



\paragraph{Certified deletion for keys or secure key leasing.}
Kitagawa and Nishimaki~\cite{AC:KitNis22} defined secret key functional encryption with secure key leasing, where we can generate a classical certificate for the deletion of functional keys. They achieved such a primitive with bounded collusion-resistance from one-way functions.
Agrawal, Kitagawa, Nishimaki, Yamada, and Yamakawa~\cite{myEPRINT:AKNYY23} defined PKE with secure key leasing, where we lose decryption capability after we return a quantum decryption key. They achieved it from standard PKE. They also extended the notion to ABE with secure key leasing and public key functional encryption with secure key leasing. They achieved them from standard ABE and public key functional encryption, respectively.
Garg et al.~\cite{EPRINT:BGGKMRR23} presented a public key functional encryption with secure key leasing scheme based on IO and injective one-way functions.
Ananth, Poremba, and Vaikuntanathan~\cite{EPRINT:AnaPorVai23} also defined the same notion as PKE with secure key leasing (they call key-revocable PKE). However, their definition differs slightly from that of Agrawal et al.~\cite{myEPRINT:AKNYY23}. They also studied key-revocable FHE and PRF. They achieved them from the LWE assumption.

\paragraph{Secure software leasing.}
Ananth and La Placa~\cite{EC:AnaLaP21} defined secure software leasing, where we lose software functionality after we return a quantum software. They achieved secure software leasing for a subclass of evasive functions from public key quantum money and the LWE assumptions. After that, several works proposed extensions, variants, and improved constructions of secure software leasing~\cite{EPRINT:ColMajPor20,TCC:BJLPS21,TCC:KitNisYam21,C:ALLZZ21}.


\section{Preliminaries}\label{sec:prel}
\subsection{Notations}\label{sec:notations}
Here we introduce basic notations we will use in this paper.

In this paper, standard math or sans serif font stands for classical algorithms (e.g., $C$ or $\algo{Gen}$) and classical variables (e.g., $x$ or $\keys{pk}$).
Calligraphic font stands for quantum algorithms (e.g., $\Gen$) and calligraphic font and/or the bracket notation for (mixed) quantum states (e.g., $\qstate{q}$ or $\ket{\psi}$).

Let $x\leftarrow X$ denote selecting an element $x$ from a finite set $X$ uniformly at random, and $y\leftarrow A(x)$ denote assigning to $y$ the output of a quantum or probabilistic or deterministic algorithm $A$ on an input $x$.
When $D$ is a distribution, $x\leftarrow D$ denotes sampling an element $x$ from $D$.
$y\seteq z$ denotes that $y$ is set, defined, or substituted by $z$.
Let $[n]\seteq \{1,\dots,n\}$. Let $\secp$ be a security parameter.
For a bit string $s\in\{0,1\}^n$, $s_i$ denotes the $i$-th bit of $s$.
QPT stands for quantum polynomial time.
PPT stands for (classical) probabilistic polynomial time. 
We say that a quantum (resp. probabilistic classical) algorithm is efficient if it runs in QPT (resp. PPT). \takashi{I added it.}
A function $f: \N \ra \R$ is a negligible function if for any constant $c$, there exists $\secp_0 \in \N$ such that for any $\secp>\secp_0$, $f(\secp) < \secp^{-c}$. We write $f(\secp) \leq \negl(\secp)$ to denote $f(\secp)$ being a negligible function. 
The trace distance between two quantum states $\rho$ and $\sigma$ denoted as $\TD(\rho,\sigma)$ is given by 
$\frac{1}{2}\norm{\rho-\sigma}_{\tr}$, where $\norm{A}_{\tr}\seteq \tr \sqrt{{\it A}^{\dagger}{\it A}}$ is the trace norm. 

\subsection{Cryptographic Tools}\label{sec:crypt_tool}
In this section, we review the cryptographic tools used in this paper.

\fuyuki{Do we need the definition of OWF?}
\ryo{I don't think so.}

\begin{definition}[Signature]\label{def:SIG}
 Let $\cM$ be a message space. A signature scheme for $\cM$ is a tuple of efficient \takashi{I added "efficient". I added it to all other definitions too.} algorithms $(\Gen,\Sign,\Vrfy)$ where:

\begin{description}
    \item[$\Gen(1^\secp) \ra (\vk,\sigk)$:] The key generation algorithm takes as input the security parameter $1^\secp$ and outputs  a verification key $\vk$ and a signing key $\sigk$. 

    \item[$\Sign(\sigk, \msg) \ra \sigma$:] The signing algorithm takes as input a signing key $\sigk$ and a message $\msg \in \cM$ and outputs a signature $\sigma$.

    \item[$\Vrfy(\vk, \msg,\sigma) \ra \top \mbox{ or } \bot$:] The verification algorithm is a deterministic algorithm that takes as input a verification key $\vk$, a message $\msg$ and a signature $\sigma$ and outputs $\top$ to indicate acceptance of the signature and $\bot$ otherwise.
\end{description}

\begin{description}
\item [Correctness:] For all $\msg \in \cM$, $(\vk,\sigk)$ in the range of $\Gen(1^\secp)$, and $\sigma\in \Sign(\sigk,\msg)$, we have $\Vrfy(\vk, \msg,\sigma) = \top$.
\end{description}	
\end{definition}

\begin{definition}[Deterministic Signature]\label{def:deterministic_SIG}
We say that a signature scheme $\SIG=(\Gen,\Sign,\Vrfy)$ is a deterministic signature if $\Sign(\sigk,\cdot)$ is a deterministic function.
\end{definition}

\begin{definition}[Public Key Encryption]\label{def:pke} \takashi{Added this definition.}
A PKE scheme is a tuple of efficient algorithms $\PKE=(\Gen,\Enc,\Dec)$.
\begin{description}
    \item[$\Gen(1^\secp) \ra (\pk,\sk)$:]
    The key generation algorithm takes the security parameter $1^\secp$ as input and outputs a public key $\pk$ and a secret key $\sk$.
    \item[$\Enc(\pk,m) \ra \ct$:]
    The encryption algorithm takes $\pk$ and a plaintext $m\in\bit$ as input,
    and outputs a ciphertext $\ct$.
    \item[$\Dec(\sk,\ct) \ra m^\prime$:]
    The decryption algorithm takes $\sk$ and $\ct$ as input, and outputs a plaintext $m^\prime \in \Ms$ or $\bot$.
    
\item[Correctness:]
For any $m\in\Ms$, we have
\begin{align}
\Pr\left[
m'\neq m
\ \middle |
\begin{array}{ll}
(\pk,\sk)\lrun \keygen(1^\secp)\\
\ct \lrun \Enc(\pk,m)\\
m'\la\Dec(\sk,\ct)
\end{array}
\right] 
\le\negl(\secp).
\end{align}

\item[Semantic Security:] For any QPT $\qA$, we have
\begin{align}
\abs{\Pr\left[\qA(\pk,\ct)=1
\ \middle |
\begin{array}{ll}
(\pk,\sk)\lrun \Gen(1^\secp)\\
\ct \lrun \Enc(\pk,0)
\end{array}
\right]
-\Pr\left[\qA(\pk,\ct)=1
\ \middle |
\begin{array}{ll}
(\pk,\sk)\lrun \keygen(1^\secp)\\
\ct \lrun \Enc(\pk,1)
\end{array}
\right]}=\negl(\secp).
\end{align}
\end{description}
\end{definition}

\newcommand{\fow}{f}
\newcommand{\msglen}{\ell}
\newcommand{\Rela}{\mathcal{R}}

\section{Signature with One-Time Unforgeability for BB84 states}\label{sec:sig_otu_bb84}
\paragraph{Definition.}
We first provide the definition of one-time unforgeability for BB84 states.

\begin{definition}[One-Time Unforgeability for BB84 states]
Let $\SIG=(\Gen,\Sign,\Vrfy)$ be a signature scheme.
We define the experiment $\expb{\SIG,\qA}{otu}{bb84}(1^\secp)$ between an adversary $\qA$ and challenger as follows.
\begin{enumerate}
           \item The challenger runs $(\vk,\sigk)\gets\Gen(1^\secp)$ and generates $x,\theta\la\bit^\secp$. The challenger generates a quantum state $\ket{\psi}$ by applying the map $\ket{m}\ket{0\ldots0}\ra\ket{m}\ket{\Sign(\sigk,m)}$ to $\ket{x}_\theta\tensor\ket{0\ldots0}$. The challenger gives $\vk$ and $\ket{\psi}$ to $\qA$.
            \item $\qA$ outputs a pair of message and signature $(x^\prime,\sigma^\prime)$ as the challenge message-signature pair to the challenger. 
            \item The experiment outputs $1$ if the followings are satisfied.
            \begin{itemize}
            \item $x_i \ne x_i'$ holds for some $i$ such that $\theta_i = 0$.
            \item $\Vrfy(\vk,x^\prime,\sigma^\prime) = 1$ holds.
            \end{itemize}
        \end{enumerate}
        We say $\SIG$ is one-time unforgeable for BB84 states if, for any QPT adversary $\qA$, it holds that
\begin{align}
\advb{\SIG,\qA}{otu}{bb84}(\secp) \seteq \Pr[\expb{\SIG,\qA}{otu}{bb84}(1^\secp)=1] = \negl(\secp).
\end{align}
\end{definition}

\paragraph{Construction.}
We construct a classical deterministic signature scheme $\SIG=(\Gen,\Sign,\Vrfy)$ satisfying one-time unforgeability for BB84 states using OWF $\fow$. 
The message space of $\SIG$ is $\bit^{\msglen}$.

\begin{description}
\item[$\Gen(1^\secp)$:] $ $
\begin{itemize}
\item Generate $u_{i,b}\la\bit^\secp$ for every $i\in[\msglen]$ and $b\in\bit$.
\item Compute $v_{i,b}\la\fow(u_{i,b})$ for every $i\in[\msglen]$ and $b\in\bit$.
\item Output $\vk:=(v_{i,b})_{i\in[\msglen],b\in\bit}$ and $\sigk:=(u_{i,b})_{i\in[\msglen],b\in\bit}$.
 \end{itemize}
\item[$\Sign(\sigk,\msg)$:] $ $
\begin{itemize}
\item Parse $(u_{i,b})_{i\in[\msglen],b\in\bit}\la\sigk$.
\item Output $(u_{i,\msg_i})_{i\in[\msglen]}$.
 \end{itemize}
\item[$\Vrfy(\vk,\msg,\sigma)$:] $ $
\begin{itemize}
\item Parse $(v_{i,b})_{i\in[\msglen],b\in\bit}\la\vk$ and $(w_i)_{i\in[\msglen]}\la\sigma$.
\item Output $\top$ if $u_{i,\msg_i}=\fow(w_i)$ holds for every $i\in[\msglen]$ and otherwise output $\bot$.
\end{itemize}
\end{description}

It is clear that $\SIG$ is a correct classical deterministic signature.
Also, we have the following theorem.

\begin{theorem}\label{thm:ots_bb84_from_owf}
Assume $\fow$ is OWF.
Then, $\SIG$ satisfies one-time unforgeability for BB84 states.
\end{theorem}
\begin{proof}
For a computational basis position $i\in[\secp]$ (that is, $i$ such that $\theta_i=0$), when we generate $\ket{\psi}$ by applying the map $\ket{m}\ket{0\ldots0}\ra\ket{m}\ket{\Sign(\sigk,m)}$ to $\ket{x}_\theta\tensor\ket{0\ldots0}$, $\ket{\psi}$ contains only $u_{i,b}$ and not $u_{i,1-b}$ if $x_i=b$.
From this fact, the one-time unforgeability of $\SIG$ directly follows from the security of $\fow$.
\end{proof}


\section{Certified Everlasting Lemmas}\label{sec:certified_everlasting_lemmas}

In this section, we first review the certified everlasting lemma by Bartusek and Khurana~\cite{EPRINT:BarKhu22}, which provides cryptography with privately verifiable deletion. Next, we present our new certified everlasting lemma, which provides cryptography with $\txPVD$.

\begin{lemma}[Certified Everlasting Lemma~\cite{EPRINT:BarKhu22}]\label{lem:ce}
Let $\setbk{\cZ_\secp(\cdot,\cdot,\cdot)}_{\secp \in \N}$ be an efficient quantum operation with three arguments: a $\secp$-bit string $\theta$, a bit $\beta$, and a quantum register $\qreg{A}$. 
For any QPT $\qB$, consider the following experiment $\widetilde{\cZ}_\secp^{\qB}(b)$ over quantum states, obtained by running $\qB$ as follows.
\begin{itemize}
\item Sample $x,\theta \chosen \zo{\secp}$ and initialize $\qB$ with $\cZ_\secp(\theta,b \xor \bigoplus_{i: \theta_i=1} x_i,\ket{x}_\theta)$.
\item $\qB$'s output is parsed as a string $x^\prime\in\zo{\secp}$ and a residual state on register $\qreg{B}$.
\item If $x_i = x^\prime_i$ for all $i$ such that $\theta_i=0$, output $\qreg{B}$, and otherwise output a special symbol $\bot$. 

Assume that for any QPT $\qA$, $\theta\in\bit^\secp$, $\beta\in\bit$, and an efficiently samplable state $\ket{\psi}^{\qreg{A,C}}$ on registers $\qreg{A}$ and $\qreg{C}$, we have
\[
\abs{\Pr[\qA(\cZ_{\secp}(\theta,\beta,\qreg{A}),\qreg{C})=1] - \Pr[\qA(\cZ_{\secp}(0^\secp,\beta,\qreg{A}),\qreg{C})=1]} \le \negl(\secp).
\]
Then, for any QPT $\qB$, we have
\[
\TD(\widetilde{\cZ}_{\secp}^{\qB}(0),\widetilde{\cZ}_{\secp}^{\qB}(1))\le \negl(\secp).
\]
\end{itemize}
\end{lemma}

\begin{remark}\label{rem:difference_from_BK}
Besides notational differences, the above lemma has the following differences from the original lemma by \cite{EPRINT:BarKhu22}.
\begin{itemize}
\item We focus on QPT adversaries $\qA$ and $\qB$, though the original lemma by \cite{EPRINT:BarKhu22} captures more general classes of adversaries.
\item The roles of computational basis position and Hadamard basis position in $\widetilde{\cZ}_\secp^{\qA_\secp^\prime}(b)$ are switched.
\end{itemize}
\end{remark}

We can upgrade~\cref{lem:ce} to a publicly verifiable one by using signatures with one-time unforgeability for BB84 state introduced in~\cref{sec:sig_otu_bb84}.

\begin{lemma}[Publicly Verifiable Certified Everlasting Lemma]\label{lem:pv_ce}
Let $\SIG=(\Gen,\Sign,\Vrfy)$ be a signature scheme satisfying one-time unforgeability for BB84 states.
Let $\setbk{\cZ_\secp(\cdot,\cdot,\cdot,\cdot,\cdot)}_{\secp \in \N}$ be an efficient quantum operation with five arguments: a verification key $\vk$ and a signing key $\sigk$ of $\SIG$, a $\secp$-bit string $\theta$, a bit $\beta$, and a quantum register $\qreg{A}$. 
For any QPT $\qB$, consider the following experiment $\widetilde{\cZ}_\secp^{\qB}(b)$ over quantum states, obtained by running $\qB$ as follows.
\begin{itemize}
\item Sample $x, \theta \leftarrow \{0, 1\}^{\secp}$ and generate $(\vk,\sigk)\la\Gen(1^\secp)$. Generate a quantum state $\ket{\psi}$ by applying the map $\ket{m}\ket{0\ldots0}\ra\ket{m}\ket{\Sign(\sigk,m)}$ to $\ket{x}_\theta\tensor\ket{0\ldots0}$. Initialize $\qB$ with $\cZ_\secp(\vk,\sigk,\theta,b \xor \bigoplus_{i: \theta_i=1} x_i,\ket{\psi})$.
\item $\qB$'s output is parsed as a pair of strings $(x^\prime,\sigma^\prime)$ and a residual state on register $\qreg{B}$.
\item If $\Vrfy(\vk,x^\prime,\sigma^\prime)=\top$, output $\qreg{B}$, and otherwise output a special symbol $\bot$. 

Assume that for any QPT $\qA$, key pair $(\vk,\sigk)$ of $\SIG$, $\theta\in\bit^\secp$, $\beta\in\bit$, and an efficiently samplable state $\ket{\psi}^{\qreg{A,C}}$ on registers $\qreg{A}$ and $\qreg{C}$, we have
\begin{align}
\abs{\Pr[\qA(\cZ_{\secp}(\vk,\sigk,\theta,\beta,\qreg{A}),\qreg{C})=1] - \Pr[\qA(\cZ_{\secp}(\vk,0^{\sigklen},\theta,\beta,\qreg{A}),\qreg{C})=1]} &\le \negl(\secp), \textrm{~~~~and}\label{eqn_pvce_sk}\\
\abs{\Pr[\qA(\cZ_{\secp}(\vk,\sigk,\theta,\beta,\qreg{A}),\qreg{C})=1] - \Pr[\qA(\cZ_{\secp}(\vk,\sigk,0^\secp,\beta,\qreg{A}),\qreg{C})=1]} &\le \negl(\secp).\label{eqn_pvce_theta}
\end{align}
Then, for any QPT $\qB$, we have
\[
\TD(\widetilde{\cZ}_{\secp}^{\qB}(0),\widetilde{\cZ}_{\secp}^{\qB}(1))\le \negl(\secp).
\]
\end{itemize}
\end{lemma}

\begin{proof}
We first define the following event $\ttForge$.
\begin{description}
\item[$\ttForge$:] For $(x^\prime,\sigma^\prime)$ output by $\qB$ in $\widetilde{\cZ}^{\qB}_\secp(b)$, the followings are satisfied.
\begin{itemize}
\item $x_i \ne x_i^\prime$ holds for some $i$ such that $\theta_i = 0$.
\item $\Vrfy(\vk,x^\prime,\sigma^\prime) = \top$ holds.
\end{itemize}
\end{description}
We also define the event $\ttForge^*$ in the same way as $\ttForge$ except that $\qB$ is initialized with $\cZ_\secp(\vk,0^{\sigklen},\theta,b \oplus \bigoplus_{i: \theta_i = 1} x_i,\ket{\psi})$.
From \cref{eqn_pvce_sk}, we have $\Pr[\ttForge]=\Pr[\ttForge^*]+\negl(\secp)$.
Also, from the one-time unforgeability for BB84 states, we have $\Pr[\ttForge^*]=\negl(\secp)$.
Thus, we obtain $\Pr[\ttForge]=\negl(\secp)$.

We define $\widehat{\cZ}_\secp^{\qB}(b)$ as the experiment defined in the same way as $\widetilde{\cZ}_\secp^{\qB}(b)$ except that the experiment outputs the register $\qreg{B}$ if and only if $x_i=x^\prime_i$ holds for all $i$ such that $\theta_i=0$.
Since $\Pr[\ttForge]=\negl(\secp)$, if $\Vrfy(\vk,x^\prime,\sigma^\prime) = \top$ holds, then $x_i = x_i^\prime$ holds for all $i$ such that $\theta_i = 0$, except negligible probability.
Thus, we have
\begin{align}
\TD(\widetilde{\cZ}_\secp^{\qB}(0),\widetilde{\cZ}_\secp^{\qB}(1)) \le \TD(\widehat{\cZ}_\secp^{\qB}(0),\widehat{\cZ}_\secp^{\qB}(1))+\negl(\secp).
\end{align}
From \cref{eqn_pvce_theta} and \cref{lem:ce}, we have $\TD(\widetilde{\cZ}_\secp^{\qB}(0),\widetilde{\cZ}_\secp^{\qB}(1)) =\negl(\secp)$.
This completes the proof.
\end{proof}


\section{Instantiations}\label{instantiations}

We can apply~\cref{lem:pv_ce} to various cryptographic primitives.
\subsection{Public-Key Encryption}
First, we show how to apply~\cref{lem:pv_ce} to PKE and obtain PKE with $\txPVD$.
\paragraph{Definition.}
We define PKE with publicly verifiable deletion (PKE-PVD).

\begin{definition}[PKE with Publicly Verifiable Deletion]\label{def:pke_pvd}
A PKE scheme with publicly verifiable deletion is a tuple of efficient algorithms $\PKEPVD=(\Gen,\qEnc,\qDec,\qDelete,\Vrfy)$.
\begin{description}
    \item[$\Gen(1^\secp) \ra (\pk,\sk)$:]
    The key generation algorithm takes the security parameter $1^\secp$ as input and outputs a public key $\pk$ and a secret key $\sk$.
    \item[$\qEnc(\pk,m) \ra (\vk,\qct)$:]
    The encryption algorithm takes $\pk$ and a plaintext $m\in\bit$ as input,
    and outputs a verification key $\vk$ and a ciphertext $\qct$.
    \item[$\qDec(\sk,\qct) \ra m^\prime$:]
    The decryption algorithm takes $\sk$ and $\qct$ as input, and outputs a plaintext $m^\prime \in \Ms$ or $\bot$.
    \item[$\qDelete(\qct) \ra \cert$:]
    The deletion algorithm takes $\qct$ as input and outputs a certification $\cert$.
    \item[$\Vrfy(\vk,\cert)\ra \top$ or $\bot$:]
    The verification algorithm is a deterministic algorithm that takes $\vk$ and $\cert$ as input, and outputs $\top$ or $\bot$.

\item[Decryption Correctness:]
For any $m\in\Ms$, we have
\begin{align}
\Pr\left[
m'\neq m
\ \middle |
\begin{array}{ll}
(\pk,\sk)\lrun \keygen(1^\secp)\\
(\vk,\qct) \lrun \qEnc(\pk,m)\\
m'\la\qDec(\sk,\qct)
\end{array}
\right] 
\le\negl(\secp).
\end{align}

\item[Verification Correctness:]
For any $m\in\Ms$, we have
\begin{align}
\Pr\left[
\Vrfy(\vk,\cert)=\bot
\ \middle |
\begin{array}{ll}
(\pk,\sk)\lrun \keygen(1^\secp)\\
(\vk,\qct) \lrun \qEnc(\pk,m)\\
\cert \lrun \qDelete(\qct)
\end{array}
\right] 
\leq
\negl(\secp).
\end{align}

\item[Semantic Security:] For any QPT $\qA$, we have
\begin{align}
\abs{\Pr\left[\qA(\pk,\vk,\qct)=1
\ \middle |
\begin{array}{ll}
(\pk,\sk)\lrun \Gen(1^\secp)\\
(\vk,\qct) \lrun \qEnc(\pk,0)
\end{array}
\right]
-\Pr\left[\qA(\pk,\vk,\qct_1)=1
\ \middle |
\begin{array}{ll}
(\pk,\sk)\lrun \keygen(1^\secp)\\
(\vk,\qct) \lrun \qEnc(\pk,1)
\end{array}
\right]}=\negl(\secp).
\end{align}
\end{description}
\end{definition}

%
%

\begin{definition}[Certified Deletion Security for PKE-PVD]\label{def:cd_security_pkepvd}
Let $\PKEPVD=(\Gen,\qEnc,\qDec,\allowbreak \qDelete,\Vrfy)$ be a PKE-PVD scheme.
We consider the following security experiment $\expb{\PKEPVD,\qA}{cert}{del}(\secp,b)$ against a QPT adversary $\qA$.
\begin{enumerate}
    \item The challenger computes $(\pk,\sk) \la \Gen(1^\secp)$, and sends $\pk$ to $\qA$.
    \item The challenger computes $(\vk,\qct) \la \qEnc(\pk,b)$, and sends $\vk$ and $\qct$ to $\qA$.
    \item At some point, $\qA$ sends $\cert$ and its internal state $\rho$ to the challenger.
    \item The challenger computes $\Vrfy(\vk,\cert)$.
    If the outcome is $\top$, the challenger outputs $\rho$ and otherwise outputs $\bot$.
\end{enumerate}
We say that $\PKEPVD$ satisfies certified deletion security if for any QPT $\qA$, it holds that
\begin{align}
\TD(\expb{\PKEPVD,\qA}{cert}{del}(\secp,0),\expb{\PKEPVD,\qA}{cert}{del}(\secp,1)) \le \negl(\secp).
\end{align}
\end{definition}

\begin{remark}
We define PKE-PVD and the certified deletion security for it as the plaintext space of PKE-PVD is $\bit$ by default.
We can generalize them into one for the plaintext space $\bit^\ell$ for any polynomial $\ell$.
Also, such a multi-bit plaintext PKE-PVD can be constructed from a single-bit plaintext PKE-PVD by the standard hybrid argument.
\end{remark}

\paragraph{Construction.}
We construct a PKE-PVD scheme $\PKEPVD=(\Gen,\qEnc,\qDec,\qDelete,\Vrfy)$ for the plaintext space $\bit$.
The building blocks are as follows.
\begin{itemize}
\item A public-key encryption scheme $\PKE=\PKE.(\Gen,\Enc,\Dec)$.
\item A deterministic signature $\SIG=\SIG.(\Gen,\Sign,\Vrfy)$.
\end{itemize}

\begin{description}
\item[$\Gen(1^\secp)$:] $ $
\begin{itemize}
\item Output $(\pk,\sk)\la\PKE.\Gen(1^\secp)$.
 \end{itemize}
\item[$\qEnc(\pk,\msg\in\bit)$:] $ $
\begin{itemize}
\item Generate $x, \theta \leftarrow \{0, 1\}^{\secp}$ and generate $(\vk,\sigk)\la\SIG.\Gen(1^\secp)$. Generate a quantum state $\ket{\psi}$ by applying the map $\ket{m}\ket{0\ldots0}\ra\ket{m}\ket{\SIG.\Sign(\sigk,m)}$ to $\ket{x}_\theta\tensor\ket{0\ldots0}$.
\item Generate $\pke.\ct\la\PKE.\Enc(\pk,(\sigk,\theta,\msg\oplus\bigoplus_{i:\theta_i=1}x_i))$.
\item Output $\vk$ and $\qct:=(\ket{\psi},\pke.\ct)$.
 \end{itemize}
 \item[$\qDec(\sk,\qct)$:] $ $
 \begin{itemize}
 \item Parse $\qct$ into a quantum states $\rho$  and classical bit string $\pke.\ct$.
 \item Compute $(\sigk,\theta,\beta)\la\PKE.\Dec(\sk,\pke.\ct)$.
 \item Apply the map $\ket{m}\ket{0\ldots0}\ra\ket{m}\ket{\SIG.\Sign(\sigk,m)}$ to $\rho$, measure the first $\secp$ qubits of the resulting state in Hadamard basis, and obtain $\bar{x}$.
 \item Output $m\la\beta\oplus\bigoplus_{i:\theta_i=1}\bar{x}_i$. 
 \end{itemize}
 \item[$\qDelete(\qct)$:] $ $
 \begin{itemize}
  \item Parse $\qct$ into a quantum state $\rho$ and a classical string $\pke.\ct$.
 \item Measure $\rho$ in the computational basis and obtain $x^\prime$ and $\sigma^\prime$.
 \item Output $(x^\prime,\sigma^\prime)$. 
 \end{itemize}
\item[$\Vrfy(\vk,\cert)$:] $ $
\begin{itemize}
\item Parse $(z^\prime,\sigma^\prime)\la\cert$.
\item Output the result of $\SIG.\Vrfy(\vk,z^\prime,\sigma^\prime)$.
\end{itemize}
\end{description}

We see that $\PKEPVD$ satisfies decryption correctness and verification correctness if $\SIG$ and $\PKE$ satisfy their correctness notions.
Also, the semantic security of $\PKEPVD$ immediately follows from that of $\PKE$.
Also, we have the following theorem.

\begin{theorem}\label{thm:PKE_PVD}
Assume $\SIG$ satisfies one-time unforgeability for BB84 states and $\PKE$ satisfies semantic security.
Then, $\PKEPVD$ satisfies certified deletion security.
\end{theorem}
\begin{proof}
We define $\cZ_\secp(\vk,\sigk,\theta,\beta,\qreg{A})$ be an efficient quantum process such that it generates $(\pk,\sk)\la\Gen(1^\secp)$ and outputs $(\pk,\vk,\Enc(\pk,(\sigk,\theta,\beta))$.
Then, from the semantic security of $\PKE$, for any QPT $\qA$, key pair $(\vk,\sk)$ of $\SIG$, $\theta\in\bit^\secp$, $\beta\in\bit$, and an efficiently samplable state $\ket{\psi}^{\qreg{A,C}}$ on registers $\qreg{A}$ and $\qreg{C}$, we have
\begin{align}
\abs{\Pr[\qA(\cZ_{\secp}(\vk,\sigk,\theta,\beta,\qreg{A}),\qreg{C})=1] - \Pr[\qA(\cZ_{\secp}(\vk,0^{\sigklen},\theta,\beta,\qreg{A}),\qreg{C})=1]} &\le \negl(\secp), \textrm{~~~~and}\\\
\abs{\Pr[\qA(\cZ_{\secp}(\vk,\sigk,\theta,\beta,\qreg{A}),\qreg{C})=1] - \Pr[\qA(\cZ_{\secp}(\vk,\sigk,0^\secp,\beta,\qreg{A}),\qreg{C})=1]} &\le \negl(\secp).
\end{align}
Then, the certified deletion security of $\PKEPVD$ follows from \cref{lem:pv_ce}.
\end{proof}

In \cref{sec:sig_otu_bb84}, we show that a deterministic signature scheme with one-time unforgeability for BB84 states is implied by OWF that is implied by a PKE scheme. 
Thus, we obtain the following theorem.
\begin{theorem}\label{thm:pke_pvd_from_pke}
Assume that there exists PKE.
Then, there exists PKE-PVD.
\end{theorem}


\subsection{Other Primitives}
By combining \cref{lem:pv_ce} with other types of primitives instead of PKE, we immediately obtain them with publicly verifiable deletion.
That is, we instantiate $\cZ_{\secp}$ in~\cref{lem:pv_ce} with these primitives.  
Formally, we have the following theorem.
\begin{theorem}\label{thm:advanced_encryption_pvd}
If there exists 
\[
Z' \in \setbk{\textrm{SKE, COM, PKE, ABE, QFHE, TRE}}, 
\]
then, there exists $Z'$ with publicly verifiable deletion.
If there exist WE and one-way functions, then there exists WE with publicly verifiable deletion. 
\end{theorem}
\begin{remark}
    We additionally assume one-way functions for the case of WE since WE is unlikely to imply one-way functions while any of $Z' \in \setbk{\textrm{SKE, COM, PKE, ABE, QFHE, TRE}}$ implies them.
\end{remark}
We omit the definitions of these primitives (with publicly verifiable deletion) and refer the reader to \cite{EPRINT:BarKhu22,myEPRINT:BarKhuPor23} for them.\footnote{The definitions in \cite{EPRINT:BarKhu22} only consider privately verifiable deletion, but it is straightforward to extend them to ones with publicly verifiable deletion.} 




\section{Making Assumptions Minimal}
In \Cref{instantiations}, we show that  $Z \in \setbk{\textrm{SKE, COM, PKE, ABE, QFHE, TRE, WE}}$ (and one-way functions in the case of $Z=\text{WE}$) imply $Z$ with publicly verifiable deletion. However, $Z$ with publicly verifiable deletion does not seem to imply either of plain (classical) $Z$ or one-way functions in general. 
In this section, we explain how to weaken the assumptions to minimal ones implied by $Z$ with publicly verifiable deletion.

First, we observe that our compiler works even if we start with $Z \in \setbk{\textrm{SKE, COM, PKE, ABE, QFHE, TRE, WE}}$ that has quantum encryption and decryption (or committing) algorithms.  
Second, we observe our compiler works  even if the underlying digital signature scheme has a quantum key generation algorithm since the quantum encryption algorithm runs it. Moreover, such a digital signature scheme with a quantum key generation algorithm that satisfies one-time security for BB84 states can be constructed from hard quantum planted problems for $\NP$ defined below. 

\begin{definition}[Hard Quantum Planted Problem for $\NP$]
A quantum polynomial-time algorithm $\qalgo{G}$ is a sampler for an $\NP$ relation $\Rela \subseteq \zo{\ast}\times\zo{\ast}$ if, for every $n$, $\qalgo{G}(1^n)$ outputs a pair $(x,w)$ such that $(x,w)\in \Rela$ with probability $1$.
We say that the quantum planted problem corresponding to $(\qalgo{G},\Rela)$ is hard if, for every QPT $\qA$, it holds that
\[\Pr[(x, \qA(x)) \in \Rela \mid (x,w)\gets \qalgo{G}(1^n)] \le \negl(n).\]
\end{definition}
It is clear that the construction in \cref{sec:sig_otu_bb84} works using hard quantum planted problems for $\NP$ instead of one-way functions where input-output pairs of the one-way function are replaced with witness-instance pairs of the $\NP$ problem if we allow the key generation algorithm of the signature scheme to be quantum. 
Combining the above observations, we obtain the following theorem.
\begin{theorem}
Assume that there exists
\[
Z \in \setbk{\textrm{SKE, COM, PKE, ABE, QFHE, TRE, WE}}
\]
with quantum encryption and decryption (or committing) algorithms 
and hard quantum planted problems for $\NP$. 
Then, there exists $Z$ with publicly verifiable deletion. 
\end{theorem}
The assumption in the above theorem is minimal since 
\begin{enumerate}
    \item $Z$ with quantum encryption and decryption (or committing) algorithms is immediately implied by $Z$ with publicly verifiable deletion by simply ignoring the deletion and verification algorithms, and
    \item Hard quantum planted problems for $\NP$ are implied by $Z$ with publicly verifiable deletion by regarding the verification key and certificate as instance and certificate, respectively. 
\end{enumerate}
\begin{remark}
In the second item above, we assume that the verification algorithm of  $Z$ with publicly verifiable deletion is a classical deterministic algorithm. If it is allowed to be a quantum algorithm, we need to consider hard quantum planted problems for $\mathsf{QCMA}$ instead of $\NP$.  The construction in \cref{sec:sig_otu_bb84} works with such problems if we allow the verification algorithm to be quantum. Thus, the minimality of the assumption holds in this setting as well. 
\end{remark}


\ifnum\anonymous=1
\else


\fi

	\ifnum\llncs=1
\bibliographystyle{bib/extreme_alpha} 
\bibliography{bib/abbrev3,bib/crypto,bib/siamcomp_jacm,bib/other}
	\else
\bibliographystyle{alpha} 
\bibliography{bib/abbrev3,bib/crypto,bib/siamcomp_jacm,bib/other}
	\fi

\ifnum\cameraready=0
	\ifnum\llncs=0
	\appendix

\else
	\newpage
	 	\appendix
 	{
	\noindent
 	\begin{center}
	{\Large SUPPLEMENTAL MATERIALS}
	\end{center}
 	}
	\setcounter{tocdepth}{2}
	 	\ifnum\noaux=1
 	\else
{\color{red}{We attached the full version of this paper as a separated file (auxiliary supplemental material) for readability. It is available from the program committee members.}}
\fi

	\setcounter{tocdepth}{1}
	\tableofcontents

	\fi
	\else
\fi

\end{document}